\documentclass[11pt]{amsart}

\usepackage{slashbox}

\usepackage{amsfonts,latexsym}
\usepackage{amsmath}
\usepackage{amstext}
\usepackage{amssymb}
\usepackage{color}
\usepackage{graphicx}
\graphicspath{{figures/}}

\newcommand{\R}{\mathbb{R}}

\newcommand{\N}{{\mathbb{N}}}

\newcommand{\eg}{{\it e.g. }}

\newcommand{\bigdot}{\bullet}

\newcommand{\maxmdj}{\displaystyle\max_{d\in D}{m_j\cdot d}}

\newtheorem{theorem}{Theorem}[section]
\newtheorem{lemma}[theorem]{Lemma}
\newtheorem{proposition}[theorem]{Proposition}

\theoremstyle{definition}
\newtheorem{definition}[theorem]{Definition}
\newtheorem{algorithm}[theorem]{Algorithm}

\newtheorem{problem}[theorem]{Problem}
\theoremstyle{remark}

\numberwithin{equation}{section}

\begin{document}

\title[Synthesis for Constrained Nonlinear Systems]
      {Synthesis for Constrained Nonlinear Systems\\ using Hybridization and Robust Controllers on Simplices}
\thanks{This work was supported by the Agence Nationale de la Recherche (VEDECY project - ANR 2009 SEGI 015 01).}


\author[Antoine Girard]{Antoine Girard}
\address{Laboratoire Jean Kuntzmann \\
Universit\'e de Grenoble \\
B.P. 53, 38041 Grenoble, France} \email{Antoine.Girard@imag.fr}

\author[Samuel Martin]{Samuel Martin}
\address{Laboratoire Jean Kuntzmann \\
Universit\'e de Grenoble \\
B.P. 53, 38041 Grenoble, France}
\email{Samuel.Martin@imag.fr}


\maketitle

\begin{abstract}
In this paper, we propose an approach to controller synthesis for a class of constrained nonlinear systems. 
It is based on the use of a hybridization, that is a hybrid abstraction of the nonlinear dynamics.
This abstraction is defined on a triangulation of the state-space where on each simplex of the triangulation, 
the nonlinear dynamics is conservatively approximated by an affine system subject to disturbances.
Except for the disturbances, this hybridization can be seen as a piecewise affine hybrid system on simplices for which appealing control synthesis
techniques have been developed in the past decade. 
We extend these techniques to handle systems subject to disturbances by synthesizing and coordinating local robust affine controllers defined on the simplices of the triangulation. 
We show that the resulting hybrid controller can be used to control successfully the original constrained nonlinear system.
Our approach, though conservative, can be fully automated and is computationally tractable.
To show its effectiveness in practical applications, we 
apply our method to control a pendulum mounted on a cart.
\end{abstract}

\section{Introduction}

Abstraction based approaches to control synthesis have become quite popular (see~\cite{tabuada2009} and the references therein).
These consist in computing a simple conservative approximation (i.e. the abstraction) of the system dynamics.
A controller synthesized using the abstraction can then be used to control the original system.
Most of these approaches use discrete abstractions since these allow to leverage controller synthesis techniques developed in the area of discrete event systems. The use of hybrid abstractions, called hybridizations, though common for verification purpose~\cite{Henzinger1998,Asarin2007,Ramdani2008,Dang2010}, has been much less explored in the context of controller synthesis. 

In this paper, we present an approach to control synthesis for a class of constrained nonlinear systems based on the use of a hybridization. 
The abstraction is defined on a partition of the state-space in simplices (i.e. a triangulation). In each of these simplices, the nonlinear dynamics is approximated by an affine system with disturbances 
where the disturbances account for the approximation errors. The resulting abstraction can be seen as a piecewise affine hybrid system on simplices subject to disturbances.
For such systems, without disturbances, automated control synthesis techniques have been developed by Habets and van Schuppen in a series of papers~\cite{Habets2001,Habets2004,Habets2006} (see also the related work in~\cite{Belta2005,Kloetzer2006,Roszak2006}). 
We extend these techniques to handle the systems subject to disturbances by synthesizing local robust affine controllers defined on the simplices of the triangulation. 
We show that by coordinating these local controllers it is then possible to control successfully the original constrained nonlinear system. A preliminary version of these results was presented in~\cite{Girard2008}. 
The technical content in the present version has been improved: 
the class of considered abstractions is more general (the disturbance is constrained in an arbitrary compact convex set); 
we synthesize hybrid controllers instead of piecewise smooth controllers, this allows us to prove the correctness of the synthesized controller more rigorously;
finally, a more challenging application example is considered.

The paper is organized as follows. In the next section, we formulate the control problem under consideration. In section 3, we introduce the notion of hybridizations and show how these can be computed for nonlinear systems. In section 4, we extend the techniques of~\cite{Habets2006,Roszak2006} for computing local
robust affine controllers on simplices. In section 5, we show how to use these local controllers and the hybridization in order to synthesize a controller for the original system. Finally, in section 6, we show the effectiveness of our approach by applying it to control a pendulum mounted on a cart.

\section{Problem Formulation}

In this paper, we consider a constrained nonlinear system of the form:
$$
\Sigma:
\left\{
\begin{array}{l}
\dot {\bf x}(t) = f( {\bf x}(t)) + g({\bf x}(t)) {\bf u}(t), \\ {\bf x}(t)\in X,\; {\bf u}(t)\in U
\end{array}\right.
$$
where the { state domain} $X \subseteq \R^n$ is the union of a finite number of compact polytopes
and the { control domain} $U\subseteq \R^m$ is a compact polytope.
We shall assume that the maps $f: X \rightarrow \R^n$ and $g: X \rightarrow \R^{n\times m}$ are
of class $\mathcal C^2$ and $\mathcal C^1$ respectively. 
Let $X_I\subseteq X$ and $X_T\subseteq X$ be compact polytopes, specifying a set of initial states and a set of target states, respectively.

\begin{definition}
A {\it hybrid controller} for system $\Sigma$ is a tuple $\mathcal C=(Q,E,\mathcal D,\mathcal G, \mathcal H)$, 
where $Q$ is a finite set of modes; 
$E\subseteq Q\times Q$ is the set of edges; 
$\mathcal D=\{D_q\subseteq X|\; q\in Q\}$ is the set of domains with $D_q$ a compact set for all $q\in Q$;
$\mathcal G=\{G_e \subseteq D_q \cap D_{q'} |\; e=(q,q') \in E \}$ is the set of guards
with $G_e$ a compact set for all $e\in E$;
$\mathcal H=\{h_q : D_q \rightarrow U|\; q\in Q \}$ is the set of continuous controllers with $h_q$ a Lipschitz continuous map for all $q\in Q$.
\end{definition}

The system $\Sigma$ controlled by $\mathcal C$ is a hybrid system, denoted $\Sigma || \mathcal C$ whose evolution is described by executions~\cite{Lygeros2003}. An \textit{execution} 
of $\Sigma || \mathcal C$ is a (finite or infinite) sequence 
of triples $(\tau_i,q_i,\mathbf x_i)_{i=0}^N$ consisting of time intervals $\tau_i=[t_i,t_{i+1}]$ (if $N$ is finite we can have
$\tau_N=[t_N,+\infty)$), modes $q_i\in Q$, and differentiable maps $\mathbf x_i:\tau_i\rightarrow X$ such that
$t_0=0$, ${\mathbf x}_0(t_0)\in X_I\cap D_{q_0}$; for all $i=0,\dots,N$ with $i< +\infty$, for all $t\in [t_i,t_{i+1})$
$$
{\mathbf x}_i(t) \in D_{q_i} \text{ and }
\dot {\mathbf x}_i(t) = f({\mathbf x}_i(t))+g({\mathbf x}_i(t)) h_{q_i}({\mathbf x}_i(t)); 
$$
and for all $i<N$,
$$
(q_i,q_{i+1})\in E, \text{ and } {\mathbf x}_i(t_{i+1})={\mathbf x}_{i+1}(t_{i+1}) \in G_{(q_i,q_{i+1})}. 
$$
The execution is said {\it infinite} if $N$ is infinite or $\tau_N=[t_N,+\infty)$, otherwise it is said {\it finite};
it is said {\it Zeno} if $N$ is infinite and 
$\lim_{i\rightarrow +\infty} t_i$ is finite. 
$\Sigma || \mathcal C$ is said to be {\it non-blocking} if all finite executions can be extended into infinite ones; it is said {\it non-Zeno} if it does not have Zeno executions.
We consider the following synthesis problem:
\begin{problem}
\label{prob:motion}
Synthesize a hybrid controller $\mathcal C$, such that $X_I \subseteq \bigcup_{q\in Q} D_q$ and the hybrid system $\Sigma || \mathcal C$ is non-blocking,
non-Zeno and all its infinite executions $(\tau_i,q_i,\mathbf x_i)_{i=0}^N$ satisfy
$$
\exists i^*\in \N,\; 0 \le i^* \le N,\; \forall i= i^*,\dots, N,\; \forall t \in \tau_i,\;  {\mathbf x}_i(t) \in X_T.
$$
\end{problem}

Let us remark that since for all $q\in Q$, $D_q\subseteq X$, 
all infinite executions of $\Sigma || \mathcal C$ also satisfy
for all $i=0,\dots,N$, for all $t \in \tau_i$, ${\mathbf x}_i(t) \in X$.
Our approach to Problem~\ref{prob:motion} consists of two main ingredients: a hybridization of the nonlinear system $\Sigma$ and 
robust continuous controllers defined on simplices. These are presented in the following sections.

\subsubsection*{Notations} In the following, for two vectors $v_1$, $v_2$, $v_1\cdot v_2$ denotes their scalar product.
For a vector $v$, $\|v\|=\sqrt{v\cdot v}$ is its Euclidean norm. For a matrix $A$, $\|A\|$ denotes its norm induced by the Euclidean norm for vectors, and $\rho(A)$ denotes
its spectral radius.

\section{Hybridization}
\label{sec:hyb}
A { hybridization} is a hybrid abstraction of a continuous dynamical system~\cite{Asarin2007}.
Conservativeness of the approximation is ensured by the introduction of disturbances.
In this paper, we shall use hybridizations
for control synthesis. 
We first describe the hybridization principle. 

\begin{definition}
$\mathcal S=\{S_p|\; p\in P\}$ is a {\it triangulation} of the domain $X$ if the following conditions hold:
for all $S_p \in \mathcal S$, $S_p$ is a full dimensional simplex of $\R^n$ (i.e. $S_p$ is the convex hull of $n+1$ affinely independent points); 
for all $S_{p},\; S_{p'} \in \mathcal S$, their intersection is the convex hull of their common vertices or empty;
$\bigcup_{p\in P} S_p = X$.
\end{definition}

We further assume that $\mathcal S$ contains 
a triangulation of $X_T$, 
$\mathcal S_T \subseteq \mathcal S$.
The main idea of the hybridization principle consists in approximating the dynamics of $\Sigma$ locally, in each simplex $S_p \in {\mathcal S}$, by an affine dynamics with disturbances of the form:
$$
\left\{
\begin{array}{l}
 \dot {\bf x}(t) = A_p{\bf x}(t) + B_p{\bf u}(t)+ a_p + {\bf w}(t),\\
{\bf u}(t) \in U,\; {\bf w}(t) \in W_p
\end{array}
\right.
$$
where the disturbance set $W_p$ is assumed to be a non empty convex compact set;
let us define $\Sigma'=\{(S_p,A_p,B_p,a_p,W_p)|\; p\in P\}$.
With the exception of the disturbance, the system $\Sigma'$ can be seen as a piecewise affine hybrid systems on simplices for which control synthesis techniques have been developed in~\cite{Habets2006}.
The disturbance is added to compensate the approximation error and guarantees the conservativeness of the approximation:  
\begin{definition}
\label{def:conservative} 
$\Sigma'$ is a {\it hybridization} of system $\Sigma$ if and only if
for all $p\in P$, 
\begin{equation}
\label{eq:bound}
\forall x\in S_p, \; \forall u\in U,\; (f(x) + g(x) {u} -A_p x - B_p u- a_p) \in W_p.
\end{equation}
\end{definition}

We now propose a method for the computation of a hybridization.
We do not discuss the computation of a triangulation $\mathcal S$ of the domain $X$.
This is a well studied problem in computational geometry for which 
efficient algorithms exist, at least in low dimensional spaces (see \eg\cite{Orourke1998}). 
In higher dimensional spaces, provided the domain can be partionned in hypercubes, a triangulation of  $X$ 
can be obtained using a simple triangulation of each hypercube as shown in~\cite{Asarin2007}.
It is to be noted that the size of the simplices of the triangulation can generally be made arbitrary small.

We now focus on the computation of the affine dynamics in simplex $S_p \in \mathcal S$.
Let $\{v_{p,0},\dots,v_{p,n}\}$ denote the $n+1$ vertices of the simplex.
Essentially, the map $f(x)$ is approximated by the affine vector field $A_p x+ a_p$ while the map $g(x)$ is approximated
by the constant matrix $B_p$. $A_p$ and $a_p$ can be obtained by interpolation of $f$ at the vertices of the simplex:
$$
\forall i\in \{0,\dots,n\},\; A_p v_{p,i} +a_p = f(v_{p,i}).
$$
Since $v_{p,0},\dots,v_{p,n}$ are affinely independent, this condition uniquely determines $A_p$ and $a_p$
which can be computed by solving a set of linear equations. 
Then, the matrix $B_p$ can be chosen as the barycentric combination of the value of $g$ 
at the vertices of the simplex:
$$
B_p=\frac{1}{n+1} \sum_{i=0}^{n} g(v_{p,i}).
$$
It remains to compute the disturbance set $W_p$ such that equation (\ref{eq:bound}) holds.
For $l\in \{1,\dots,n\}$, let $f_l(x)$ denote the $l$-th component of the vector $f(x)$, $g_{l\bigdot}(x)$ denote the $l$-th line $l$ of the matrix $g(x)$, $a_{p,l}$ denote the $l$-th element of $a_p$, $A_{p,{l\bigdot}}$ and $B_{p,l\bigdot}$ denote the $l$-th lines of 
$A_p$ and $B_p$. The set of disturbances $W_p$ can be chosen as
$W_p =  [\underline\mu_{p,1},\overline\mu_{p,1}] \times \dots \times [\underline\mu_{p,n},\overline\mu_{p,n}]$
where for $l \in \{1,\dots,n\}$, 
$$
\underline\mu_{p,l}= \min_{x\in S_p,u\in U} \left( f_l(x)+g_{l\bigdot}(x)u -A_{p,l\bigdot} x -B_{p,l\bigdot} u -a_{p,l} \right)
$$
and
$$
\overline\mu_{p,l}= \max_{x\in S_p,u\in U} \left( f_l(x)+g_{l\bigdot}(x)u -A_{p,l\bigdot} x -B_{p,l\bigdot} u -a_{p,l} \right).
$$
Then, it is clear that equation (\ref{eq:bound}) holds. The following result shows that for the hybridization defined above, the disturbance set $W_p$
can be made arbitrarily small provided the triangulation of the domain is sufficiently fine.
\begin{proposition}
\label{pro:error2} Let $\Sigma'=\{(S_p,A_p,B_p,a_p,W_p)|\; p\in P\}$ be the hybridization defined above.
Let $\textsf{H}(f_l)$ denote the Hessian matrix of $f_l$ and $\textsf{J}(g_{l\bigdot})$ denote the Jacobian matrix of $g_{l\bigdot}$. Let $\delta_p$ denote the {\it diameter} of the simplex $S_p$: $\delta_p=\max \{\|x-x'\||\; {x,x' \in S_p}\}$. Then, $W_p \subseteq [-\mu_{p,1},\mu_{p,1}] \times \dots \times [-\mu_{p,n},\mu_{p,n}]
$
where for $l \in \{1,\dots,n\}$, 
$$
\mu_{p,l}=\alpha_{p,l} \delta_p^2 + \beta_{p,l} \delta_p \max_{u \in U} \|u\|
$$
with
$
\displaystyle{\alpha_{p,l} = \frac{1}{2} \max_{x\in S_p} \rho(\textsf{H}(f_l)(x))\; \text{ and }\;  \beta_{p,l} =  \max_{x\in S_p} \|\textsf{J}(g_{l\bigdot})(x)\|}.
$
\end{proposition}

\begin{proof}
It is clearly sufficient to show that for all $x\in S_p$, for all $l\in \{1,\dots,n\}$
$$
 |f_l(x) -A_{p,l\bigdot} x -a_{p,l}| \le \alpha_{p,l} \delta_p^2 \text{ and }
 \|g_{l\bigdot}(x) - B_{p,l\bigdot}\| \le \beta_{p,l} \delta_p.
$$
A proof of the first inequality can be found in~\cite{Waldron1998}.
The second inequality is derived as follows. From the mean value theorem for vector-valued function we have for all $x_1$, $x_2$ in $S_p$
\begin{eqnarray*}
\| g_{l\bigdot}(x_1) - g_{l\bigdot}(x_2)\| \le  \max_{x\in S_p} \|\textsf{J}(g_{l\bigdot})(x)\| 
\|x_1-x_2\| \le \max_{x\in S_p} \|\textsf{J}(g_{l\bigdot})(x)\|  \delta_p.
\end{eqnarray*}
Then, it follows that for all $x\in S_p$,
\begin{eqnarray*}
\|g_{l\bigdot}(x) - B_{p,l\bigdot}\|& \le & \|g_{l\bigdot}(x) - \frac{1}{n+1} \sum_{i=0}^{n} g_{l\bigdot}(v_{p,i})\|
\le \frac{1}{n+1} \sum_{i=0}^{n} \| g_{l\bigdot}(x) - g_{l\bigdot}(v_{p,i})\| \\
&\le &\frac{1}{n+1} \sum_{i=0}^{n}  \max_{x\in S_p} \|\textsf{J}(g_{l\bigdot})(x)\|  \delta_p
=  \max_{x\in S_p} \|\textsf{J}(g_{l\bigdot})(x)\|  \delta_p.
\end{eqnarray*}
\end{proof}

We can see that the bound on the disturbance depends linearly in the diameter of the simplex.
This means that to reduce the bounds by factor $1/2$ it is necessary to consider $2^n$ times more simplices in the trinagulation.
Hopefully, as we shall see on a practical example, it is not always necessary to consider very fine partitions for our approach to be successful.

The rest of the paper is devoted to solving Problem~\ref{prob:motion}. 
Our approach is inspired by the work presented in~\cite{Habets2006} for piecewise affine hybrid systems on simplices.

\section{Robust Controllers on Simplices}

We need to extend several techniques developed for the class of affine systems on simplices to the class of systems with disturbances. The results are stated without proofs which are straightforward
adaptations of the proofs in~\cite{Habets2006,Roszak2006} and can be found in~\cite{Girard2008}. Let $S$ be a simplex of $\R^n$, we denote $v_0, \dots , v_n$ and $F_0,\dots, F_n$ the vertices and the facets of $S$
with the convention that $F_j$ is the facet opposite to vertex $v_j$. 
$m_0,\dots,m_n$ denote the outward unit normal vectors of the facets of $S$.

\subsection{Affine systems with disturbances}

We consider the following autonomous affine system with disturbances~:
\begin{equation}
\label{eq:aff}
\dot {\bf x}(t) = A{\bf x}(t) + a + {\bf w}(t), \; {\bf x}(t) \in \R^n,\; {\bf w}(t) \in W
\end{equation}
where the disturbance set $W$ is a non empty convex compact set. 
It is sufficient to assume that the disturbance ${\bf w}$ is a continuous map.
We say that a trajectory {\bf x} of (\ref{eq:aff}), starting in $S$, {\it exits} $S$ at time $T\ge 0$, 
if there exists $\varepsilon >0$ such that
$$
\forall t \in [0,T],\; {\bf x}(t) \in S \text{ and } \forall t\in (T,T+\varepsilon),\; {\bf x}(t) \notin S .
$$
We shall say that a facet $F_j$ is {\it blocked} if 
$$
\forall x\in F_j,\; m_j\cdot (Ax+a) \le -\maxmdj.
$$
The fact that a facet is blocked can be characterized using only the value of the vector field at the vertices:
\begin{proposition}
\label{pro:blocked} The facet $F_j$ is blocked if and only if 
$$
\forall i\in \{0,\dots,n\},\; i\ne j, \;
m_j\cdot (Av_i+a) \le -\maxmdj.
$$
\end{proposition}

The following proposition characterizes the points where the trajectories may exit the simplex.
\begin{proposition}\label{pro:blocked2} If a trajectory {\bf x} of (\ref{eq:aff}), starting
in $S$, exits $S$ at time $T$, then ${\bf x}(T)$ belongs to a facet of $S$ that is not blocked.
\end{proposition}

We now give necessary and sufficient conditions such that all the trajectories of (\ref{eq:aff}), starting
in $S$, exit $S$ in finite time.
\begin{proposition}
\label{pro:exit1}
All the trajectories of (\ref{eq:aff}) starting in $S$, exit $S$ in finite time if and only if
$$
\exists \xi \in \R^n,\; \forall x\in S,\; \xi \cdot (Ax+a) > -\min_{w\in W}\xi\cdot w.
$$
\end{proposition}

The previous property can again be characterized using only the value of the affine vector field at the vertices of the simplex:
\begin{proposition}
\label{pro:exit2}
All the trajectories of (\ref{eq:aff}) starting in $S$, exit $S$ in finite time if and only if
$$
\exists \xi \in \R^n,\; \forall i\in \{0,\dots,n\},\; \xi \cdot (Av_i+a) > -\min_{w\in W}\xi\cdot w.
$$
\end{proposition}

\subsection{Robust controller synthesis}
\label{sec:con}
Let us now consider an affine control systems with disturbances on simplices of the form:
$$
\dot {\bf x}(t) = A{\bf x}(t) + B{\bf u}(t)+ a + {\bf w}(t), \; {\bf x}(t) \in \R^n,\; {\bf u}(t) \in U,\; {\bf w}(t) \in W
$$
where the control domain $U\subseteq \R^p$ is a convex compact polytope and the disturbance set $W$ is a non empty convex compact set. The disturbance ${\bf w}$ is assumed to be continuous.
We consider the following control problem:

\begin{problem}
\label{prob:exit} Consider a subset of indices $\mathcal J \subseteq \{0,\dots,n\}$,
and the associated subset of facets $\mathcal F = \{F_j|,\; j\in \mathcal J\}$, 
design an affine feedback controller $h: \R^n \rightarrow \R^m$, $h(x)=Kx+k$, such that
for all $x\in S$, $h(x)\in U$ and for the autonomous affine system with disturbances
\begin{equation}
\label{eq:affcon}
 \dot {\bf x}(t) = A{\bf x}(t) + Bh({\bf x}(t))+ a + {\bf w}(t),\; {\bf x}(t) \in \R^n,\; {\bf w}(t) \in W\; 
\end{equation}
all facets that are not in $\mathcal F$ are blocked and all trajectories of (\ref{eq:affcon})
starting in $S$ 
exits $S$ in finite time.
\end{problem}

We denote $u_0,\dots,u_n\in U$ the values of the controller at the vertices of $S$:
$$
u_i=h(v_i)=K v_i + k,\; i\in \{0,\dots,n\}.
$$
Since $v_0,\dots,v_n$ are affinely independent, $u_0,\dots,u_n$ uniquely determine the matrix $K$ and the vector $k$.
Moreover since $u_0,\dots,u_n\in U$ and for all $x\in S$, $h(x)$ is a convex combination of $u_0,\dots,u_n$ it follows that
$h(x)\in U$, for all $x\in S$.
At the vertices of the simplex $S$, the value of the vector field of (\ref{eq:affcon})
 is given by
$
A v_i + Bh(v_i) +a = Av_i +B u_i +a$.
In the following we characterize suitable values of $u_0,\dots,u_n$ with the understanding that these allows the computation of the
affine controller $h$. The following result is a direct consequence of Propositions~\ref{pro:blocked} and~\ref{pro:exit2}.

\begin{proposition}
\label{pro:exit3}
For $i\in \{0,\dots,n\}$, let us consider the following convex polytopes:
$$
U_i
=\left\{u\in U \left|
m_j\cdot(Av_i +B u +a) \le -\maxmdj, \;
\forall j \in \{0,\dots,n\}\setminus( \mathcal J\cup\{ i\} )
\right.
\right\}.
$$
An affine feedback controller $h$ solves Problem~\ref{prob:exit} if and only if for all $i\in \{0,\dots,n\}$, $u_i \in U_i$ and 
$$
\exists \xi \in \R^n,\; \forall i\in \{0,\dots,n\},\; \xi \cdot (Av_i+Bu_i+a) > -\min_{w\in W}\xi\cdot w.
$$
\end{proposition}

Let us discuss the computation of input values $u_0,\dots,u_n$ satisfying the conditions of the previous proposition. 
We denote by $\overline U_i$ the set of vertices of $U_i$.
\begin{proposition}
\label{pro:exit}
There exists an affine controller $h$ solving Problem~\ref{prob:exit} if and only if there exists $u_0\in \overline U_0,\dots,u_n\in \overline U_n$ such that
$$
\exists \xi \in \R^n,\; \forall i\in \{0,\dots,n\},\; \xi \cdot (Av_i+Bu_i+a) > -\min_{w\in W}\xi\cdot w.
$$
\end{proposition}

A controller solving Problem~\ref{prob:exit} can therefore be synthesized by computing the vertices of the polytopes $U_0,\dots,U_n$ and then looking for a suitable combination of these vertices. 

\section{Nonlinear Control Synthesis}

We now describe a solution to Problem~\ref{prob:motion} based on the results described in the previous sections. 
Let $\Sigma'$ be a hybridization of $\Sigma$. Our approach is based on the coordination
of robust affine controllers defined on the simplices of the triangulation $\mathcal S$. 

\subsection{Invariance controller} We start by synthesizing a controller in order to render the set $X_T$ of target states invariant.
For all simplices $S_p\in {\mathcal S}_T$ of the triangulation of $X_T$, 
we denote by $\textit{vertices}(S_p)$ its set of vertices
and by $\textit{external-facets}(S_p,X_T)$ the set of facets of $S_p$
that are included in the boundary of the polytope $X_T$. For a given facet $F\in \textit{external-facets}(S_p,X_T)$, we denote by $m_F$ its outward unit normal vector and  by $\textit{vertices}(F)$ its set of vertices.

\begin{problem}[Invariance control]
\label{prob:invariance}
Synthesize a Lipschitz continuous controller $h_{T}: X_T \rightarrow U$, such that for all $x_0 \in X_T$,
the solution of 
\begin{equation}
\label{eq:feed}
\dot {\bf x}(t) = f( {\bf x}(t)) + g({\bf x}(t)) h_{T}( {\bf x}(t)),\; {\bf x}(0)=x_0
\end{equation}
is defined for all $t\in \R^+$.
\end{problem}

Let us remark that since $h_T$ is defined only on $X_T$, this implies that ${\bf x}(t) \in X_T$ for all $t\in \R^+$.
We will search the controller $h_T$ as a continuous piecewise affine map defined on the triangulation ${\mathcal S}_T$
(i.e. $h_T$ is affine on each simplex of the triangulation). $h_T$ is uniquely determined by its value
at the vertices of the simplices of ${\mathcal S}_T$. If all these
values belong to $U$ then so does $h_T(x)$, for all $x\in X_T$.
Also, it can be shown~\cite{Girard2004} that $h_{T}$, defined this way,
is Lispchitz on $X_T$.
Then,
from the standard characterization of invariant sets~\cite{Aubin91} we have the following result:
\begin{lemma}\label{lem:invariant} $h_T: X_T \rightarrow U$ solves Problem~\ref{prob:invariance} if and only
if for all $S_p\in {\mathcal S}_T$, for all $F\in \textit{external-facets}(S_p,X_T)$,
for all $x\in F$, $m_F\cdot ( f(x) + g(x) h_T(x) ) \le 0$.
\end{lemma}

We now characterize the values of $h_T$ at the vertices of the simplices of $\mathcal S_T$ that 
result in controllers solving Problem~\ref{prob:invariance}:
\begin{proposition}
\label{pro:invert}
If for all $S_p\in {\mathcal S}_T$, for all $F\in \textit{external-facets}(S_p,X_T)$, for all $v \in \textit{vertices}(F)$  
\begin{equation}
\label{eq:cond1}
m_F\cdot (A_p v + B_p h_T(v) +a_p)  \le -\max_{w\in W_p} m_F \cdot w
\end{equation}
then $h_T$ solves  Problem~\ref{prob:invariance}. 
\end{proposition}

\begin{proof} Let  $S_p\in {\mathcal S}_T$, $F\in \textit{external-facets}(S_p,X_T)$, and $x\in F$, then
\begin{eqnarray*}
m_F\cdot ( f(x) + g(x) h_T(x) )& =& m_F\cdot (A_p x + B_p h_T(x) +a_p)\\
&& +m_F\cdot ( f(x) + g(x) h_T(x) -A_p x - B_p h_T(x) -a_p).
\end{eqnarray*}
Since $\Sigma'$ is a hybridization of $\Sigma$, and since $h(x)\in U$, it follows from
Definition~\ref{def:conservative} that $( f(x) + g(x) h_T(x) -A_p x - B_p h_T(x) -a_p) \in W_p$. Therefore,
$$
m_F\cdot ( f(x) + g(x) h_T(x) ) \le m_F\cdot (A_ px + B_p h_T(x) +a_p) +\max_{w\in W_p} m_F \cdot w.
$$
Since $A_p x + B_p h_T(x) +a_p$ is an affine map on $F$, we have by convexity and by (\ref{eq:cond1}) that 
$$
m_F\cdot (A_p x + B_p h_T(x) +a_p) +\max_{w\in W_p} m_F \cdot w \le 0.
$$
Then, from Lemma~\ref{lem:invariant}, we have that the controller $h_T$ solves  Problem~\ref{prob:invariance}. 
\end{proof}

Let us remark that the computation of the controller $h_T$ involves only finding values of $h_T$, satisfying a set of linear inequalities, at the vertices of the triangulation.
This can be done efficiently using linear programming. For simplicity of the presentation, we assumed that we render the whole target set $X_T$ invariant; actually, as far as Problem~\ref{prob:motion} is concerned, it is sufficient to render a subset $X_T'\subseteq X_T$ invariant. This can be done a similar way. 

\subsection{Reachability controller} We now describe how to synthesize a hybrid controller solving Problem~\ref{prob:motion}. The proposed controller essentially drives the trajectories of the system through a sequence of simplices ending in $X_T$. 
This is done by coordinating robust affine controllers defined in section~\ref{sec:con}. 

Let $\mathcal S'$ be a subset of the triangulation $\mathcal S$ and let $S_p\in \mathcal S$,
we denote by $\textit{common-facets}(S_p,\mathcal S')$ the subset of facets of the simplex $S_p$
that are also facets of a simplex in $\mathcal S'$. We denote by $\textit{adjacent}(\mathcal S')$
the subset of simplices that are adjacent to a simplex in $\mathcal S'$: $S_p \in \textit{adjacent}(\mathcal S')$
if and only if $S_p \in \mathcal S \setminus \mathcal S'$ and $\textit{common-facets}(S_p,\mathcal S')\ne \emptyset$.
The synthesis of the hybrid controller $\mathcal C$ can be done using the following algorithm:

\begin{algorithm}\label{algo:reach}
Let $Q:=\{q_T\}$, $D_{q_T}:=X_T$, $h_{q_T}$ solve Problem~\ref{prob:invariance}, $E:=\emptyset$.\\
Let $k:=0$, $\mathcal S_{-1}:=\emptyset$, $\mathcal S_0:=\mathcal S_T$.\\
While $\mathcal S_{k}\ne\mathcal S_{k-1}$:
\begin{itemize}
\item $\mathcal S_{k+1}:=\mathcal S_k$.
\item For all $S_p \in \textit{adjacent}(\mathcal S_k)$,
if there exists an affine feedback controller $h_p$ which solves Problem~\ref{prob:exit} with
$
(S,\mathcal F,A,B,a,W)=(S_p, \textit{common-facets}(S_p,\mathcal S_k),A_p,B_p,a_p,W_p),
$
then let $\mathcal S_{k+1}:=\mathcal S_{k+1}\cup \{S_p\}$, 
$Q:=Q\cup \{p\}$, $D_p:=S_p$, $h_p$ solution of Problem~\ref{prob:exit} and:
\begin{itemize}
\item For $S_{p'}\in \mathcal S_k\setminus \mathcal S_T$, if $\textit{common-facets}(S_p,\{S_{p'}\})\neq \emptyset$ 
then $E:=E\cup \{(p,p')\}$, and $G_{(p,p')}:=\textit{common-facets}(S_p,\{S_{p'}\})$,

\item If $\textit{common-facets}(S_p,\mathcal S_T)\neq \emptyset$ then $E:=E\cup \{(p,q_T)\}$,
$G_{(p,q_T)}:=\textit{common-facets}(S_p,\mathcal S_T)$.
\end{itemize}

\item $k:=k+1$.
\end{itemize}
Return $\mathcal C=(Q,E,\mathcal D,\mathcal G, \mathcal H)$ where $\mathcal D=\{D_q|\; q\in Q\}$,
$\mathcal G=\{G_e|\; e\in E\}$, $\mathcal H=\{h_q|\; q\in Q\}$.
\end{algorithm}

The algorithm essentially applies dynamic programming:
$\mathcal S_k$ consists of the simplices from which we can drive the system to $X_T$, using  robust affine controllers,
by following a sequence of at most $k$ simplices; at each iteration we add to $\mathcal S_{k+1}$ the simplices
in which we can drive all the trajectories to facets that are shared with some simplices in $\mathcal S_k$.
Let us remark that the Algorithm~\ref{algo:reach} necessarily stops after $K < |\mathcal S|$ iterations and each iteration involves solving a number of linear programs that is at most linear in the number of simplices.
Then, the worst-case time complexity of the algorithm is quadratic in the number
of simplices. However, it is to be noted that the number of simplices in the triangulation is exponential in the dimension of the state-space. The algorithm is therefore exponential in the dimension of the system. This is expected
for an approach based on abstractions defined on a partition of the state-space.

\begin{theorem} Let $\mathcal C$ be computed by Algorithm~\ref{algo:reach};
if $X_I \subseteq \bigcup_{q\in Q} D_q$ then  $\mathcal C$ solves Problem~\ref{prob:motion}.
\end{theorem}

\begin{proof} 
Let us consider a finite execution $(\tau_i,q_i,\mathbf x_i)_{i=0}^N$ of $\Sigma || \mathcal C$.
If $q_N=q_T$, then since $h_{q_T}$ solves Problem~\ref{prob:invariance}, $\mathbf x_N$ can be extended on $[t_N,+\infty)$.
If $q_N\ne q_T$, then $q_N \in P$, let $k\le K$ be the smallest index such that $S_{q_N} \in \mathcal S_k$.
Since $\Sigma'$ is a hybridization of $\Sigma$, the continuous dynamics 
in mode $q_N$ is: 
\begin{eqnarray*}
\dot{\bf x}(t) &=& f({\bf x}(t)) + g({\bf x}(t)) h_{q_N}({\bf x}(t)) \\
                &=& A_{q_N}{\bf x}(t) + B_{q_N} h_{q_N}({\bf x}(t)) + a_{q_N} +{\bf w}(t),\; {\bf w}(t) \in W_{q_N}
\end{eqnarray*}
for an obvious particular value of the disturbance ${\bf w}(t)$.
By construction of $h_{q_N}$, all trajectories of the linear system above exits $S_{q_N}$ in finite time
and by Proposition~\ref{pro:blocked2}, $\mathbf x_N$ can be extended on $[t_N,t'_{N+1}]$
where ${\mathbf x}_N(t'_{N+1})$ belongs to a facet $F_{q_N}$ of $S_{q_N}$ that is not blocked for the linear system above.
By construction of the guards, there exists an edge $e=(q_N,q_{N+1})$, such that $F_{q_N}\subseteq G_e$.
Hence, the finite execution can be extended into an execution of the form $(\tau_i,q_i,\mathbf x_i)_{i=0}^{N+1}$. Moreover, either $q_{N+1}=q_T$ or we have $S_{q_{N+1}}\in \mathcal S_{k-1}$.
From the previous discussion, and using a simple induction, we can show that every finite
executions of $\Sigma || \mathcal C$ can be extended 
into infinite ones, by taking a finite number of discrete transitions (at most $K$) before reaching the final mode $q_T$
in which continuous evolution is enabled until time goes to infinity. Hence, $\Sigma || \mathcal C$ is non-blocking.
Since all executions have only a finite number of discrete transitions, $\Sigma || \mathcal C$ is necessarily non-Zeno.
Finally, all executions of $\Sigma || \mathcal C$ reach the final mode $q_T$ and stays there forever, since $D_{q_T}=X_T$,  Problem~\ref{prob:motion} is solved.
\end{proof}

Let us remark that our approach for solving Problem~\ref{prob:motion} is clearly conservative.
Our algorithm may fail to solve the problem even though a suitable controller exists. There are
several sources of conservatism. The first one is due to the use of a hybridization. Proposition~\ref{pro:error2}
suggests that this can be reduced by using a finer triangulation at the price of an increased computational effort.
The other sources of conservatism are inherent to the approach developed in~\cite{Habets2006}.
However, this conservatism allows us to synthesize controllers that are correct by design by a
fully automated method which is computationally effective
as shown in the following section.

\section{Example}

In this section, we apply our approach for controlling a pendulum mounted on a cart. 
This system has been considered in~\cite{Reissig2011} for illustrating controller synthesis 
using discrete abstractions.
The dynamics of the system is described by
$$
\left\{
\begin{array}{llll}
\dot{\bf x}_1(t)& = & {\bf x}_2(t) &\\
\dot{\bf x}_2(t)& = & -\sin({\bf x}_1(t)) - {\bf u}(t) \cos({\bf x}_1(t)) &,\; {\bf u}(t)\in [-\alpha,\alpha] 
\end{array}
\right.
$$
where ${\bf x}_1(t)$ is the angle between the pendulum and the downward vertical, ${\bf x}_2(t)$ is the angular velocity. The control input ${\bf u}(t)$ is the acceleration of the cart and is assumed to be bounded in absolute value by $\alpha$. The goal is to bring the pendulum from downward vertical position (${\bf x}_1(t)=0$) to upward vertical position (${\bf x}_1(t)=\pi$) and to keep it there. This can be formulated as Problem~\ref{prob:motion} where $X=[0,2\pi] \times [-\pi,\pi]$, $X_I=\{(0,0)\}$ and
$X_T$ is a neighborhood of $(\pi,0)$. Let us remark that since ${\bf x}_1(t)$ represents an angle then if the trajectory exits the domain $X$ with ${\bf x}_1(t)=2\pi$, it re-enters immediately with ${\bf x}_1(t)=0$ and the same velocity. This is important when defining the adjacency relations between simplices of the triangulation of $X$.


We use a regular triangulation defined on a partition of $X$ into squares of length $\pi/N$ where $N\in \N$ is a parameter. Each of these squares is then partionned into two triangles: this can be done using a line of slope $1$ or $-1$ giving two possible orientations for the simplices.
The target set $X_T$ is given by
$$
X_T=\{(x_1,x_2)\in X|\; 
\begin{array}{lll} {\frac{(N-1)\pi}{N}} \le x_1 \le \frac{(N+1)\pi}{N},& \frac{-\pi}{N} \le x_2 \le \frac{\pi}{N}, &
\frac{(N-1)\pi}{N} \le x_1+x_2 \le \frac{(N+1)\pi}{N} \end{array} \}.
$$
The orientation of the simplices is left free at the beginning and is determined on the fly by Algorithm~\ref{algo:reach}.
The hybridization is also computed on the fly. Let us remark that only the second component of the vector field is nonlinear, therefore the disturbance set $W_p$ is of the form
$W_p=\{0\}\times [\underline \mu_{p,2},\overline \mu_{p,2}]$ where the bounds 
given in Section~\ref{sec:hyb} can be computed explicitly. When computing the continuous controllers, to choose among the possible values at the vertices of the simplex, we try to minimize the variation of the control map within the simplex. In particular, if a constant controller
solves the problem then such a controller will be used.

Algorithm~\ref{algo:reach} was implemented in Matlab, and ran on a standard desktop equipped with CPU Pentium 4 (3.20 GHz).
We report in Table~\ref{tab} performances of the algorithm for varying number of simplices and set of inputs.
We can check that the computational costs increases as the constraints on the inputs become tighter. 
Also, the complexity with respect to the number of simplices, estimated experimentally on successful cases is roughly 
in $O(|\mathcal S|^{3/2})$ which is polynomial but better than the theoretical worst-case complexity which is quadratic.
Compared to the approach presented in~\cite{Reissig2011}, where the same example is considered, the approach seems 
to have quite similar performances: reported computation times are of the same order and the estimated complexity
appears to be polynomial in the number of discrete states as well.

\begin{table}
\begin{center}
\begin{tabular}{|c|c|c|c|c|}
\hline
\backslashbox{$U$}{$N$, $|\mathcal S|$} & 10, 800 & 20, 3200 & 30, 7200 & 40, 12800 \\
\hline
$[-2,2]$ & 6 - No & 176 - Yes  &  591 - Yes & 1263 - Yes\\
\hline
$[-3,3]$ & 6 - No & 102 - Yes  &  330 - Yes & 631 - Yes\\
\hline
$[-4,4]$ & 14 - Yes & 90 - Yes &  251 - Yes & 531 - Yes\\
\hline
\end{tabular}
\caption{\label{tab} Run times in seconds and success of the synthesis for varying number of simplices and
set of inputs. }
\end{center}
\end{table}

In Figure~\ref{fig2}, we show the controller obtained for $U=[-4,4]$ and $N=10$ (i.e. $|\mathcal S|=800$),
as well as the trajectory of the controlled system. The pendulum 
starts balancing in the anti-clockwise direction to accumulate some energy, then it moves in the clockwise direction until it reaches the target.
This nontrivial example shows that our approach can be used for automated controller synthesis for nonlinear systems.

\begin{figure}[!t]
\begin{center}
\vspace{-0.5cm}
\includegraphics[scale=0.78]{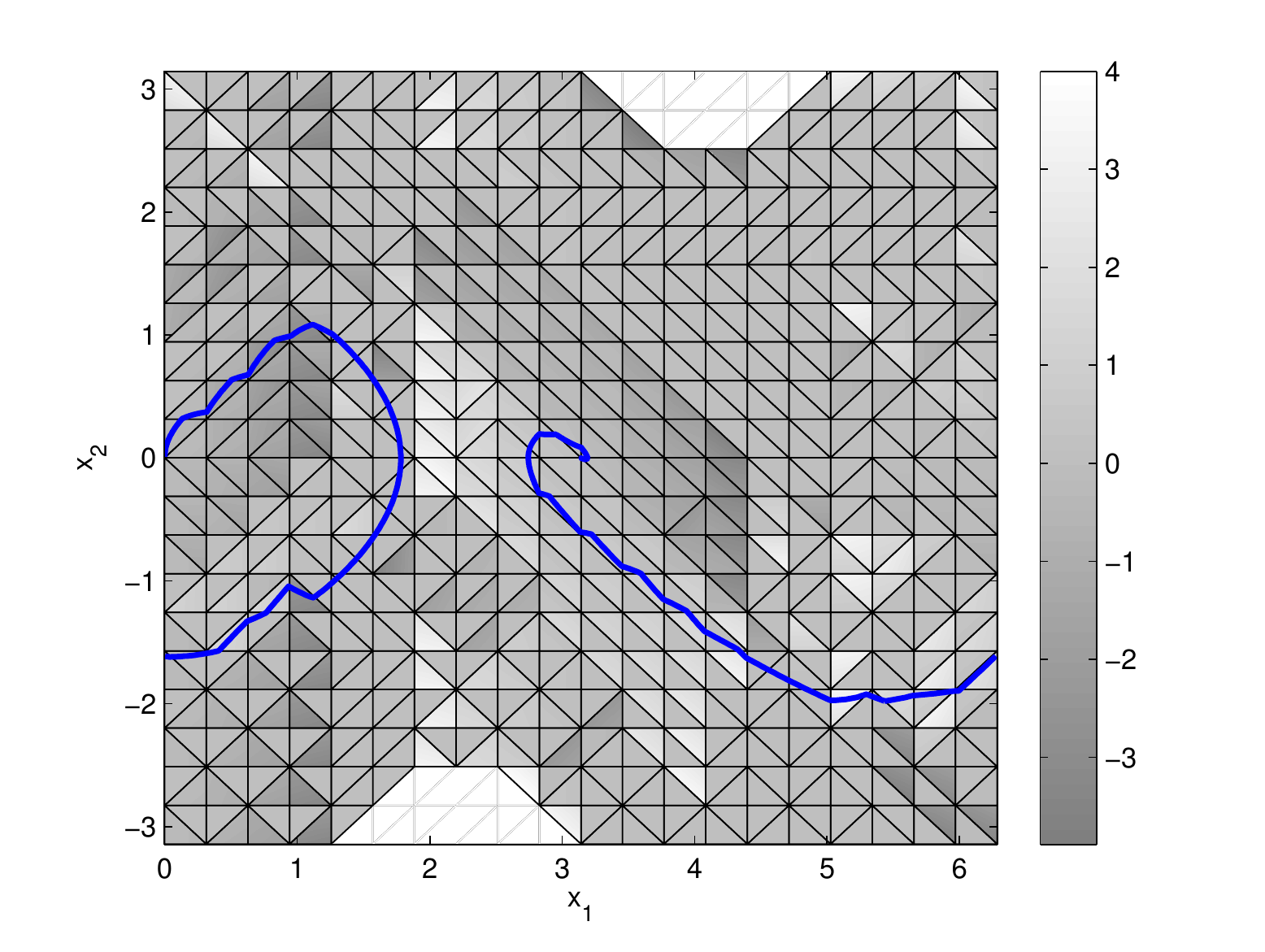}
\vspace{-0.7cm}
\caption{\label{fig2} Controller and trajectory of the pendulum on a cart for $\alpha=4$ and $N=10$. Levels of gray indicates the value of the control input constrained in $[-4,4]$.}
\end{center}
\end{figure}

\section{Conclusion}

In this paper, we presented an algorithmic approach to controller synthesis for constrained nonlinear
systems. Our technique is based on two main ingredients, namely a hybridization and robust affine controllers
on simplices. Though conservative, our method can be fully automated and we showed that it is effective on an example.
Our method should probably be reserved to small-dimensional systems as the number of simplices in the triangulation
explodes when the dimension grows.

There are several possible extensions for this work. First, instead of using piecewise affine hybridization defined 
on a triangulation,
one could use piecewise multi-affine hybridization~\cite{Asarin2007} defined on a partition of the state domain in
hypercubes as an extension of~\cite{Belta2006} would allow us to synthesize 
local controllers on hypercubes. Second, by extending the approach presented in~\cite{Fainekos2005,Kloetzer2006},
one could solve more complex control problems such as those specified in linear temporal logic (LTL). Finally,
the class of nonlinear system we consider could be extended, by considering systems of the form
$$
\dot{\bf x}(t) =f({\bf x}(t),{\bf u}(t)), \; {\bf x}(t)\in X,\; {\bf u}(t)\in U.
$$
Following~\cite{Dumas2005}, a piecewise affine hybridization could be computed by triangulating the domain $X\times U$.

\bibliographystyle{biblio/IEEEtranS} 
\bibliography{biblio/references}

\end{document}